\documentclass[english,conference,final,a4paper]{IEEEtran}
\usepackage[latin9]{inputenc}
\usepackage{babel}

\makeatletter

\usepackage[T1]{fontenc}
\usepackage[latin9]{inputenc}
\usepackage{amsthm}
\usepackage{amsmath}
\usepackage{amssymb}
\usepackage{graphicx}
\usepackage{dsfont}
\usepackage{cite}
\usepackage{amsmath}
\usepackage{array}
\usepackage{multirow}
\usepackage{caption}
\usepackage{color}

\theoremstyle{plain}

\theoremstyle{plain}

\theoremstyle{plain}
\newtheorem{thm}{\protect\theoremname}
\theoremstyle{plain}
  
\theoremstyle{definition}

\theoremstyle{definition}

\theoremstyle{definition}
\newtheorem{rem}{\protect\remarkname}

\makeatother
  
\usepackage{babel} 
\providecommand{\lemmaname}{Lemma}
\providecommand{\propositionname}{Proposition}
\providecommand{\theoremname}{Theorem}
\providecommand{\corollaryname}{Corollary} 
\providecommand{\definitionname}{Definition}
\providecommand{\assumptionname}{Assumption}
\providecommand{\remarkname}{Remark}

\newcommand{\overbar}[1]{\mkern 1.25mu\overline{\mkern-1.25mu#1\mkern-0.25mu}\mkern 0.25mu}

\newcommand{\openone}{\mathds{1}}

\newcommand{\Hg}{\mathrm{Hypergeometric}}
\newcommand{\TV}{\mathrm{TV}}
\newcommand{\dTV}{d_{\mathrm{TV}}}

\newcommand{\Bi}{\mathrm{Binomial}}

\newcommand{\Sdif}{S_{\mathrm{dif}}}
\newcommand{\Seq}{S_{\mathrm{eq}}}
\newcommand{\sdif}{s_{\mathrm{dif}}}
\newcommand{\seq}{s_{\mathrm{eq}}}

\newcommand{\Bernoulli}{\mathrm{Bernoulli}}

\newcommand{\pebar}{\overbar{P}_{\mathrm{e}}}

\newcommand{\pe}{P_{\mathrm{e}}}

\newcommand{\Xv}{\mathbf{X}}

\newcommand{\Yv}{\mathbf{Y}}

\newcommand{\Uc}{\mathcal{U}}

\newcommand{\EE}{\mathbb{E}}
\newcommand{\PP}{\mathbb{P}}

\usepackage{float} 
\usepackage{subfig}

\floatstyle{ruled}
\newfloat{algorithm}{tbp}{loa}
\providecommand{\algorithmname}{Algorithm}
\floatname{algorithm}{\protect\algorithmname}

\makeatletter
\newcommand{\manuallabel}[2]{\def\@currentlabel{#2}\label{#1}}
\makeatother

\newcommand{\Veq}{V_{\mathrm{eq}}}
\newcommand{\Vdif}{V_{\mathrm{dif}}}

\newcommand{\veq}{v_{\mathrm{eq}}}

\begin{document}

\title{Converse Bounds for Noisy Group Testing \\  with Arbitrary Measurement Matrices} 
 
\author{
 \IEEEauthorblockN{Jonathan Scarlett and Volkan Cevher}
  \IEEEauthorblockA{Laboratory for Information and Inference Systems (LIONS) \\
    \'Ecole Polytechnique F\'ed\'erale de Lausanne (EPFL) \\
    Email: \{jonathan.scarlett,volkan.cevher\}@epfl.ch}
} 

\maketitle

\begin{abstract}
    We consider the group testing problem, in which one seeks to identify a subset of defective items within a larger set of items based on a number of noisy tests.  While matching achievability and converse bounds are known in several cases of interest for i.i.d.~measurement matrices, less is known regarding converse bounds for arbitrary measurement matrices.  We address this by presenting two converse bounds for arbitrary matrices and general noise models.  First, we provide a strong converse bound ($\PP[\mathrm{error}] \to 1$) that matches existing achievability bounds in several cases of interest.  Second, we provide a weak converse bound ($\PP[\mathrm{error}] \not\to 0$) that matches existing achievability bounds in greater generality.
\end{abstract}

\section{Introduction}

The group testing problem consists of determining a small subset of ``defective'' items within a larger set of items $\{1,\dotsc,p\}$. This problem has a history in areas such as medical testing and fault detection, and has regained significant attention with following new applications in areas such as communication protocols \cite{Ant11}, pattern matching \cite{Cli10}, and database systems \cite{Cor05}, and new connections with compressive sensing \cite{Gil08,Gil07}.

Let the items be labeled as $\{1,\dotsc,p\}$, and let $S$ be the subset of defective items.  We consider a general group testing model in the observation $Y$ associated with a single test is randomly generated according to
\begin{equation}
    \PP[Y=y\,|\,X=x,S=s] = P_{Y|X_S}(y|x_s) = P_{Y|V_S}(y|v_s), \label{eq:model} \\
\end{equation}
where
\begin{equation}
    V_S := \sum_{i \in S} X_i
\end{equation}
counts the number of defective items in the test, and where the measurement vector $X = (X_1,\dotsc,X_p) \in \{0,1\}^p$ indicates which items are included in the test.  While our techniques allow for arbitrary finite output alphabets, we focus on the binary case $Y\in\{0,1\}$ for concreteness.  In the noiseless setting, we simply have $Y = \openone\{V_S > 0\}$.  Additive modulo-2 noise models of the form $Y = \openone\{V_S > 0\} \oplus Z$ are also common, but \eqref{eq:model} is more general, permitting other forms of dependence on $V_S$ such as that of dilution noise \cite{Ati12}.
 
The goal is to recover $S$ based on a number $n$ of independent non-adaptive tests, with the $i$-th measurement vector being $X^{(i)}$ and the $i$-th observation being $Y^{(i)}$.  We henceforth let $\Xv$ denote the $n \times p$ matrix whose $i$-th row is $X^{(i)}$, and let $\Yv$ be the $n$-dimensional binary vector whose $i$-th entry is $Y^{(i)}$.  We consider a fixed number $k$ of defective items, and assume that the support set $S$ is uniform over the subsets of $\{1,\dotsc,p\}$ with cardinality $k$.  For a fixed measurement matrix $\Xv$, the error probability is given by 
\begin{equation}
    \pe(\Xv) = \PP[\hat{S} \ne S], \label{eq:pe}
\end{equation}
where $\hat{S}$ is the estimate of $S$ based on $\Xv$ and $\Yv$, and the probability is with respect to the randomness in $S$ and $\Yv$.

The information-theoretic limits of this problem have been studied for decades (e.g., see \cite{Mal78,Mal13}), and have recently become increasingly well-understood \cite{Ati11,Tan14,Laa14,Sca15,Sca15b,Ald15}.  In particular, an exact asymptotic threshold is known in several cases of interest when we consider the error probability $\pebar := \EE[\pe(\Xv)]$ \emph{averaged over an i.i.d.~Bernoulli matrix $\Xv$ with $\PP[X_{ij} = 1] = \nu/k$} ($\nu > 0$).  Specifically, in a broad range of scaling regimes with $k = o(p)$, we have  $\pebar \to 0$ if \cite{Sca15b}
 \begin{equation}
     n \ge \max_{\ell=1,\dotsc,k} \frac{\ell \log \frac{p}{\ell}}{ I(X_{\sdif} ; Y | X_{\seq} ) } (1+\eta), \label{eq:ach}
 \end{equation}
and $\pebar \to 1$ if
\begin{equation}
    n \le \max_{\ell=1,\dotsc,k} \frac{\ell \log \frac{p}{\ell} }{ I(X_{\sdif} ; Y | X_{\seq} ) } (1-\eta). \label{eq:cnv}
\end{equation}
In both of these equations, $(\sdif,\seq)$ denotes an arbitrary partition of a fixed defective set $s$ with $|\sdif| = \ell$ (see \cite{Sca15b} for further intuition), and the mutual information is with respect to the independent random vectors $(X_{\sdif},X_{\seq})$ of sizes $(\ell,k-\ell)$ containing independent $\Bernoulli(\nu/k)$ entries, and the model in \eqref{eq:model} with $s = \sdif \cup \seq$.

The main goal of this paper is to obtain variants of the converse bound with an additional optimization over $\nu$, i.e.,
\begin{equation}
    n \le \min_{\nu \in [0,k]} \max_{\ell=1,\dotsc,k} \frac{\ell \log \frac{p}{\ell} }{ I(X_{\sdif} ; Y | X_{\seq} ) } (1-\eta), \label{eq:cnv_gen}
\end{equation}
in the case of \emph{arbitrary} measurement matrices, rather than i.i.d.~measurement matrices.\footnote{ Although the measurement matrix $\Xv$ may be arbitrary, our final results are still written in terms of random vectors $X_{\sdif}$ and $X_{\seq}$ having independent Bernoulli entries.  These are not \emph{directly} related to $\Xv$ itself.}  We briefly mention some existing works in this direction:
\begin{itemize}
    \item For the noiseless setting $Y = \openone\{V_S > 0\}$, the threshold in \eqref{eq:cnv_gen} simplifies to $\big(k\log_2\frac{p}{k}\big) (1 - \eta)$ \cite{Sca15b}, and the converse holds for arbitrary matrices by the so-called \emph{counting bound} \cite{Bal13,Joh15}.
    \item For the symmetric noise model $Y = \openone\{V_S > 0\} \oplus Z$ with $Z \sim \Bernoulli(\rho)$ for some $\rho \in (0,1)$, the threshold in \eqref{eq:cnv} simplifies to $\frac{k\log_2\frac{p}{k}}{\log 2 - H_2(\rho)} (1 - \eta)$ \cite{Sca15b}, where $H_2(\rho) := -\rho\log \rho - (1-\rho)\log(1-\rho)$ is the binary entropy function in nats. Moreover, the converse remains valid for arbitrary matrices. This can be proved by combining the analysis of \cite{Sca15b} with a simple symmetry argument on the information-density random variables, or can alternatively be obtained from a non-asymptotic bound given in \cite{Joh15}.
    \item For general noise models, a \emph{weak converse} statement corresponding to $\ell = k$ (i.e., $\seq = \emptyset$) is known for arbitrary matrices \cite{Mal78} (i.e., showing $\pe(\Xv) \not\to 0$ as opposed to the strong converse $\pe(\Xv) \to 1$). 
    \item After the initial preparation of this work, we learned that a result similar to our second one (Theorem \ref{thm:weak} below) was presented in the Russian literature \cite[pp.~630-631]{Mal13}, giving a weak converse for the case $\seq \ne \emptyset$.  However, the proof techniques appear to be significantly different, and the focus therein is on the case that $k$ does not scale with $p$, in contrast with our work.
\end{itemize}

\subsection{Contributions}

In this paper, we prove a strong converse corresponding to $\ell = k$ for arbitrary matrices, and we prove a weak converse for all $\ell = 1,\dotsc,k$.  Note that the former of these is of interest since $\ell = k$ often achieves the maximum in \eqref{eq:cnv_gen}; this is true for the noiseless model and the symmetric noise model \cite{Mal78,Sca15b}, and our numerical investigations suggest that it is also the case when $P_{Y|V_S}$ corresponds to passing $\openone\{V_S > 0\}$ through a Z-channel \cite{Cov01}.  However, there are known cases where only smaller values of $\ell$ achieve the maximum \cite{Mal80}.
 
\subsection{Notation}
 
 We write $\Xv_{S}$ to denote the submatrix of $\Xv$ containing the columns indexed by $S$.  The complement with respect to the set $\{1,\dotsc,p\}$ is denoted by $(\cdot)^c$, and similarly for $X_S^{(i)}$.   For a given joint distribution $P_{XY}$, the corresponding marginal distributions are denoted by $P_{X}$ and $P_{Y}$, and similarly for conditional marginals (e.g.,~$P_{Y|X}$).  We use usual notations for the entropy and mutual information (e.g. $H(X)$, $I(X;Y|Z)$).  We make use of the standard asymptotic notations $O(\cdot)$, $o(\cdot)$, $\Theta(\cdot)$, $\Omega(\cdot)$ and $\omega(\cdot)$.  We define the function $[\cdot]^+ = \max\{0,\cdot\}$, and write the floor function as $\lfloor\cdot\rfloor$.  The function $\log$ has base $e$.  The total variation (TV) distance between two probability mass functions is written as $\dTV(P,Q)$.
 
 \section{Strong Converse for $\seq = \emptyset$}
 
Our first main result is as follows.

\begin{thm} \label{thm:strong}
    Consider any observation model $P_{Y|V_S}$, and define $I_s^* := \max_{\nu \in [0,k]} I(X_s;Y)$, where $X_{s}$ has i.i.d.~$\Bernoulli(\nu/k)$ entries. For any sequence of measurement matrices $\Xv$ (indexed by $p$), we have 
    \begin{equation}
        \pe(\Xv) \ge 1 - O\bigg( \frac{1}{n (I_s^*)^2} \bigg)
    \end{equation}
    provided that
    \begin{equation}
        n \le \frac{\log {p \choose k} }{ I_s^* } (1-\eta),
    \end{equation}
    for arbitrarily small $\eta > 0$.
\end{thm}
 
\begin{rem}
    Typically in the case that $\seq = \emptyset$ we have $I_s^* = \Theta(1)$, and hence the remainder term $O\big( \frac{1}{n (I_s^*)^2} \big)$ behaves as $O\big(\frac{1}{n}\big)$, in which case this lower bound on the error probability yields the strong converse statement $\pe(\Xv) \to 1$.
\end{rem}
 
\begin{proof}[Proof of Theorem \ref{thm:strong}]
     Let $\Xv \in \{0,1\}^{n \times p}$ be a fixed measurement matrix.  The analysis of \cite{Sca15b} shows that
     \begin{multline}
         \pe(\Xv) \ge \sum_{s} \frac{1}{{p \choose k}} \PP\bigg[ \sum_{i=1}^n \log \frac{P_{Y|X_S}(Y^{(i)}|X_s^{(i)})}{Q_Y(Y^{(i)})} \le \log{{p \choose k}}  \\ + \log \delta_1 \,\Big|\, \Xv, S=s \bigg] - \delta_1,
     \end{multline}
     where $Q_Y$ is an arbitrary auxiliary output distribution.  Specifically, this was proved in \cite{Sca15b} for the case that $\Xv$ is i.i.d.~and $Q_Y$ is an induced output distribution, but the proof reveals this more general form.
     
     Letting $\mu_n(s)$ and $\sigma_n^2(s)$ denote the mean and variance of $\sum_{i=1}^n \log \frac{P_{Y|X_S}(Y^{(i)}|X_s^{(i)})}{Q_Y(Y^{(i)})}$ for a given defective set $s$, we obtain from Chebyshev's inequality that $\pe \ge 1 - \sum_{s} \frac{1}{{p \choose k}} \frac{\sigma_n(s)^2}{(n \Delta I_s^*)^2} - \delta_1$ provided that $\log{{p \choose k}} + \log \delta_1 \le \mu_n(s) + n\Delta I_s^*$ for all $s$; here $\Delta \in (0,1)$ is arbitrary for now.
     
    The mean is directly computed as
    \begin{align}
        \mu_n(s) &= \sum_{i=1}^n \sum_{y} P_{Y|X_S}(y|X_s^{(i)}) \log \frac{P_{Y|X_S}(y|X_s^{(i)})}{Q_Y(y)} \\
            & = \sum_{i=1}^n \sum_{y} P_{Y|V_S}(y|V_s^{(i)}) \log \frac{P_{Y|V_S}(y|V_s^{(i)})}{Q_Y(y)} \\
            & = n\sum_{v_s,y} P^{(s)}_{V_S}(v_s)P_{Y|V_S}(y|v_s)\log\frac{P_{Y|V_S}(y|v_s)}{Q_Y(y)}. \label{eq:mu_bound3}
    \end{align}
    where $V_s^{(i)} := \sum_{j \in s} X_j^{(i)}$, and $P^{(s)}_{V_S}$ is the empirical distribution of $V_S$ across the $n$ tests for a given choice of $s$.  Choosing $Q_Y$ to be the unique capacity-achieving output distribution of the ``channel'' $P_{Y|V_S}$, it follows from \eqref{eq:mu_bound3} and a well-known saddlepoint result on the mutual information \cite[Thm.~4.4]{Pol14} that, for all sets $s$ having cardinality $k$, we have
    \begin{equation}
        \mu_n(s) \le nI_s^*, \label{eq:mu_bound}
    \end{equation}
    where $I_s^*$ is defined in the theorem statement.
    
    We claim that the corresponding variance behaves as 
    \begin{equation}
        \sigma_n^2(s) = O(n). \label{eq:var_bound}
    \end{equation}
    This was shown for the case that $Q_Y$ equals an induced output distribution in \cite[App.~A]{Sca15}, but the analysis reveals that the same holds true for any $Q_Y$ such that $\min\{Q_Y(0),Q_Y(1)\}$ is bounded away from zero.
    
    The proof of Theorem \ref{thm:strong} is concluded by combining \eqref{eq:mu_bound} and \eqref{eq:var_bound} with the above-mentioned application of Chebyshev's inequality, and choosing $\eta$ such that $1 - \eta < \frac{1}{1 + \Delta}$, Since $\Delta$ can be arbitrarily small, the same is true of $\eta$.
\end{proof}

\section{Weak Converse for $\seq \ne \emptyset$}

Our second main result is as follows.
\begin{thm} \label{thm:weak}
    For any observation model $P_{Y|V_S}$ and sequence of measurement matrices $\Xv$ (indexed by $p$), we have $\pe(\Xv) \not\to 0$ provided that
    \begin{equation}
        n \le \max_{\ell=1,\dotsc,k} \min_{\nu \in [0,k]} \frac{ {p-k+\ell \choose \ell} }{ I(X_{\sdif} ; Y | X_{\seq} ) + \Delta_{\ell} } (1-\eta) \label{eq:cnv2}
    \end{equation}
    for some $\eta > 0 $, where
    \begin{equation}
        \Delta_{\ell} = C_0\frac{\ell(k-\ell)}{p} \max\Big\{1, \log\frac{p}{\ell(k-\ell)} \Big\}
    \end{equation}
    for some universal constant $C_0$, and the mutual information is with respect to the pair $(X_{\sdif},X_{\seq})$ having i.i.d.~$\Bernoulli(\nu/k)$ entries, along with \eqref{eq:model}.
\end{thm}

\begin{rem}
    The remainder term $\Delta_{\ell}$ is typically (but not always) dominated by the mutual information; for example, if the mutual information is $\Theta(1)$ then this is true when $k = O(p^{\theta})$ for some $\theta < \frac{1}{2}$, regardless of the value of $\ell$. 
\end{rem}

\begin{rem}
    The min-max ordering in \eqref{eq:cnv2} is the opposite of that in \eqref{eq:cnv_gen}, thus making it a potentially weaker threshold.  However, in the proof we also show that the threshold can be improved to 
    \begin{equation}
        \min_{\Uc,P_U, P_{X|U}} \max_{\ell} \frac{\log{ p-k+\ell \choose \ell }}{ I(X_{\sdif};Y|X_{\seq},U) + \Delta_{\ell}  } (1-\eta), \label{eq:cnv_modified}
    \end{equation}
    thus recovering the correct min-max ordering, but with an additional random variable $U$ on a finite alphabet $\Uc$.  This threshold can be shown to be achievable (hence establishing that \eqref{eq:cnv_modified} is a tight bound) in a broad range of scaling regimes using i-non-i.d. coding: Fix a sequence $(u_1,\cdots,u_n)$ with empirical distribution $P_U$, and then generate the $i$-th row according to an i.i.d.~Bernoulli distribution $P_{X|U}(\cdot|u_i)$ whose parameter $\nu$ may depend on $u_i$.  The achievability analysis then follows that in \cite{Sca15b}.
    
    We have chosen to state the theorem in terms of the weakened threshold \eqref{eq:cnv2} since it bears a stronger resemblance to the more familiar threshold \eqref{eq:cnv_gen}, and since we are not aware of any cases in which there is a gap between the two.
\end{rem}

\begin{proof}[Proof of Theorem \ref{thm:weak}]

    The proof is given in four steps.
    \vspace{1mm}
    
    {\em Step 1 (Fano's Inequality):}
    The starting point of our analysis is a necessary condition for $\pe(\Xv) \to 0$ based on Fano's inequality and a genie argument, which follows directly from the analysis of \cite{Ati12} (see also \cite[Sec.~III-D]{Sca15}).  Specifically, fixing $\ell = 1,\dotsc,k$ and letting the \emph{revealed indices of $S$} (denoted $\Seq$) be uniform on the set of subsets $\{1,\dotsc,p\}$ of size $k-\ell$, and letting the \emph{non-revealed indices of $S$} (denoted $\Sdif$) be uniform on the set of subsets of $\{1,\dotsc,p\} \backslash \Seq$ of size $\ell$, it is necessary that
    \begin{equation}
        1 \ge \frac{ \log{{p-k+\ell \choose \ell }} }{ I(\Sdif;\Yv | \Seq) } (1+o(1)). \label{eq:Fano_initial}
    \end{equation}
    We upper bound the mutual information by writing
    \begin{align}
        I(\Sdif;\Yv | \Seq) 
            &\le \sum_{i=1}^n I(\Sdif;Y^{(i)}|\Seq) \label{eq:MI_sum1} \\
            &= \sum_{i=1}^n I(\Vdif^{(i)};Y^{(i)}|\Veq^{(i)}), \label{eq:MI_sum2}
    \end{align}
    where \eqref{eq:MI_sum1} is a standard property for independent observation models \cite[Eq.~(7.96)]{Cov01}, and \eqref{eq:MI_sum2} follows by defining  $(\Vdif^{(i)},\Veq^{(i)})$ to count the number of defective items in the $i$-th test at the non-revealed and revealed indices, and recalling from \eqref{eq:model} that $Y$ depends on the defective set $S = \Sdif \cup \Seq$ only through $V_S := \sum_{i \in S} X_i$.
    
    {\em Step 2 (Approximate Distributions by Binomials):}
    We proceed by showing that the pairs $(\Vdif^{(i)},\Veq^{(i)})$ have a distribution which is ``close enough'' to a product of Binomial distributions with the same probability parameter.  Since the defective set is uniformly random, the joint distribution of each pair $(\Vdif^{(i)},\Veq^{(i)})$ (and hence the mutual information $I(\Vdif^{(i)};Y^{(i)}|\Veq^{(i)})$) only depends on the number of non-zeros in the $i$-th row $X^{(i)}$ of $\Xv$, which we denote by $m^{(i)}$.
    
    Before proceeding, we recall that the Hypergeometric($k$, $m$, $p$) distribution counts the number of ``special items'' obtained when sampling $k$ items from a population of $p$ items \emph{without replacement}, $m$ of which are labeled as special.  A random variable with this distribution has probability mass function $P_H(i) = \frac{ {m \choose i} {p - m \choose k - i} }{ {p \choose k} }$.  Of course, sampling \emph{with replacement} simply gives the Binomial($k$, $m/p$) distribution.
    
    We have the following:
    \begin{enumerate}
        \item Recalling that $\Seq$ is uniform on the ${p \choose k-\ell}$ sets having cardinality $k-\ell$, the number of ones at the revealed indices is distributed as
        \begin{equation}
            \Veq^{(i)} \sim \Hg(k-\ell,m^{(i)},p).
        \end{equation}
        We approximate this by the Binomial random variable 
        \begin{equation}
            \Veq^{(i)} \sim \Bi\Big(k-\ell,\frac{m^{(i)}}{p}\Big).
       \end{equation}
         Specifically, denoting the corresponding distributions by $P_{\Veq}$ and $P_{\Veq'}$ respectively (omitting the superscripts $(\cdot)^{(i)}$), the total variation distance between the two satisfies \cite{Soo96}
        \begin{equation}
            \dTV(P_{\Veq},P_{\Veq'}) \le \frac{k-\ell-1}{p-1} = O\bigg( \frac{k-\ell}{p} \bigg). \label{eq:d1}
       \end{equation}
       We denote this upper bound by $\delta_1$.
        \item Suppose that we condition on some value $\veq$ of $\Veq^{(i)}$.  Recalling that $(\Sdif|\Seq=\seq)$ is uniform on the ${p-k+\ell \choose \ell}$ possible realizations, we have
        \begin{multline}
            (\Vdif^{(i)} \,|\, \Veq^{(i)} = \veq) \\ \sim \Hg(\ell, m^{(i)} - \veq,p - k + \ell).
        \end{multline}
        We approximate this by the conditional distribution
        \begin{equation}
            (\Vdif^{\dagger(i)} \,|\, \Veq^{(i)} = \veq) \sim \Bi\Big(\ell,\frac{m^{(i)} - \veq}{p - k + \ell}\Big),
        \end{equation}
        which we further approximate by the unconditional distribution
        \begin{equation}
            \Vdif^{\prime(i)} \sim \Bi\Big(\ell,\frac{m^{(i)} }{p}\Big).
        \end{equation}
          Specifically, the corresponding distributions satisfy \cite{Soo96}
        \begin{multline}
            \dTV(P_{\Vdif}(\cdot|\veq),P_{\Vdif^\dagger}(\cdot|\veq)) \\ \le \frac{\ell-1}{p-k+\ell-1} = O\bigg( \frac{\ell}{p} \bigg), \label{eq:d21}
       \end{multline}
       and (proved in the Appendix)
       \begin{equation}
           \dTV(P_{\Vdif^\dagger}(\cdot|\veq),P_{\Vdif'}) = O\bigg( \frac{\ell(k-\ell)}{p} \bigg) \label{eq:d22}
       \end{equation}
       uniformly in $m^{(i)}$ and $\veq$.  Denoting these bounds by $\delta_{2,1}$ and $\delta_{2,2}$, we obtain from the triangle inequality that
       \begin{equation}
           \dTV(P_{\Vdif}(\cdot|\veq),P_{\Vdif'}) \le \min\{1,\delta_{2,1}+\delta_{2,2}\} =: \delta_2,
       \end{equation}
       where the upper bound of one is trivial.
    \end{enumerate}

    {\em Step 3 (Infer Bounds on the Mutual Informations)}
    
    Next, we formalize the statement that if two joint distributions are close in TV distance, their (conditional) mutual informations are also close. Using the above definitions of $(\Vdif,\Veq)$, $(\Vdif',\Veq')$ and $(\delta_1,\delta_2)$, we have the following:
    \begin{enumerate}
        \item We prove in the Appendix that
        \begin{equation}
            \big| I(\Vdif;Y|\Veq) - I(\Vdif;Y|\Veq')\big| \le \delta_1 \log 2. \label{eq:Ibound1}
        \end{equation}
        \item We also prove in the Appendix that
        \begin{equation}
            \big| I(\Vdif;Y|\Veq') - I(\Vdif';Y|\Veq')\big| \le \delta_2 \log\frac{4}{\delta_2}. \label{eq:Ibound2}
        \end{equation}
        In fact, we show that the logarithmic term can usually be improved to a constant and sometimes even $o(1)$; see Remark \ref{rem:log_factor}.  We focus on the slightly looser bound \eqref{eq:Ibound2} for the sake of simplicity.
        \item Combining these with \eqref{eq:d1}, \eqref{eq:d21} and \eqref{eq:d22} gives
        \begin{multline}
            \big| I(\Vdif;Y|\Veq) - I(\Vdif';Y|\Veq')\big| \\ = O\bigg( \frac{\ell(k-\ell)}{p} \max\Big\{1, \log\frac{p}{\ell(k-\ell)} \Big\} \bigg). \label{eq:multi_letter}
        \end{multline}
    \end{enumerate}
    Substituting \eqref{eq:multi_letter} into \eqref{eq:Fano_initial} and \eqref{eq:MI_sum2}, and maximizing over $\ell$, we obtain the necessary condition
    \begin{equation}
        n \ge \max_{\ell} \frac{\log{ p-k+\ell \choose \ell }}{ \frac{1}{n}\sum_{i=1}^n I(\Vdif^{'(i)};Y^{(i)}|\Veq^{'(i)}) + \Delta_{\ell}  } (1+o(1)) \label{eq:Step2_bound}
    \end{equation}
    where $\Delta_{\ell} = O\big( \frac{\ell(k-\ell)}{p} \max\{1, \log\frac{p}{\ell(k-\ell)} \} \big)$.

    {\em Step 4 (Form a Single-letter Expression)}

    By defining a random variable $U$ equiprobable on $\{1,\dotsc,n\}$, we can write the average in the denominator of \eqref{eq:Step2_bound} as 
    \begin{equation}
        \frac{1}{n}\sum_{i=1}^n I(\Vdif^{'(i)};Y^{(i)}|\Veq^{'(i)}) = I(\Vdif';Y|\Veq',U),
    \end{equation}
    where the conditional distributions of $\Vdif'$ and $\Veq'$ given $U=i$ are independent Binomial random variables with $(\ell,k-\ell)$ trials and a common parameter $\frac{m^{(i)}}{p}$.  Thus, the overall bound becomes
    \begin{equation}
        n \ge \max_{\ell} \frac{\log{ p-k+\ell \choose \ell }}{ I(\Vdif';Y|\Veq',U) + \Delta_{\ell}  } (1+o(1)) \label{eq:Step3_bound}
    \end{equation}
    Upper bounding the right-hand side by maximizing over $P_{U}$ and $P_{X|U}$ yields \eqref{eq:cnv_modified}; once again, since the output depends on the measurement vector $X$ only through $\sum_{i \in S} X_i$, we can safely replace the Binomial random variables $(\Vdif',\Veq')$ by the corresponding i.i.d.~Bernoulli vectors $(X_{\seq},X_{\sdif})$ in the mutual information.  Further weakening \eqref{eq:cnv_modified} by swapping the min-max ordering yields \eqref{eq:cnv2}, thus concluding the proof of Theorem \ref{thm:weak}.
\end{proof}

\section{Conclusion}

We have provided two converse bounds for noisy group testing with arbitrary measurement matrices.  Our first result strengthens an existing result \cite{Mal78} to obtain a strong converse statement $\pe(\Xv) \to 1$, and our second result provides a (weak) converse with a potentially improved threshold.  In several cases, these converse bounds are known to be achievable using i.i.d.~matrices when $k$ scales sufficiently slowly compared to $p$ \cite{Sca15b}, and thus our results support the use of such matrices in these regimes.  In contrast, it is known that i.i.d.~matrices can be suboptimal in other settings, such as the linear scaling $k = \Theta(p)$ \cite{Ald15}.  In such cases, there may be room to improve the converse bounds presented in this paper.

Another direction for future work is to determine to what extent our bounds remain valid in the case of adaptive group testing, where each test can be designed based on past observations.  Some work in this direction is given in \cite{Joh15}, but the most conclusive results therein are limited to symmetric noise.

\appendix

\section{Appendix}

\subsection{Proof of \eqref{eq:d22}}

Recall that we are considering the TV distance between $(\Vdif^\dagger | \Veq = \veq) \sim \Bi\big(\ell,\frac{m^{(i)} - \veq}{p - k + \ell}\big)$ and $\Vdif' \sim \Bi\big(\ell,\frac{m^{(i)} }{p}\big)$.   We define the difference between the two binomial parameters as
\begin{equation}
   \Delta := \frac{m^{(i)}}{p} - \frac{m^{(i)} - \veq}{p - k + \ell}.
\end{equation}
By a simple asymptotic expansion and the fact that $\veq \in [0,k-\ell]$, this satisfies
\begin{equation}
    \Delta = O\bigg( \frac{k-\ell}{p} \bigg)
\end{equation}
uniformly in $m^{(i)}$ and $\veq$.  Moreover, the bound for comparing Binomial distributions in \cite[Eq.~(16)]{Roo01} states that 
\begin{equation}
    \dTV(P_{\Vdif^\dagger}(\cdot|\veq),P_{\Vdif'}) \le c\sqrt{\eta}(1+\sqrt{2\eta})e^{2\eta},
\end{equation}
where $c = (2\pi)^{1/4}e^{1/24}2^{-1/2} $ and $\eta = \Delta^2\ell(\ell+2) = O(\Delta^2\ell^2)$.  This upper bound behaves as $O(\sqrt{\eta}) = O\big( \frac{\ell(k-\ell)}{p} \big)$ whenever $\eta = O(1)$, thus establishing \eqref{eq:d22}.  If $\eta = \Omega(1)$, then \eqref{eq:d22} is trivial anyway, since it gives $\frac{\ell(k-\ell)}{p} = \Omega(1)$, but an upper bound of $1$ always holds.

\subsection{Proof of \eqref{eq:Ibound1}}

We obtain \eqref{eq:Ibound1} by writing
\begin{align}
    & \big| I(\Vdif;Y|\Veq) - I(\Vdif;Y|\Veq')\big| \nonumber \\
    &\qquad = \bigg| \sum_{\veq} \big( P_{\Veq}(\veq) - P_{\Veq'}(\veq) \big) I(\Vdif;Y|\veq)\bigg| \label{eq:Ibound_pf1} \\
    &\qquad \le \sum_{\veq} \big|  P_{\Veq}(\veq) - P_{\Veq'}(\veq)\big|   \log 2 \label{eq:Ibound_pf2}  \\
    &\qquad = d_{\TV}(P_{\Veq}, P_{\Veq'}) \log 2, \label{eq:Ibound_pf3} 
\end{align}
where \eqref{eq:Ibound_pf2} holds since the mutual information is upper bounded by $\log 2$ with binary outputs.

\subsection{Proof of \eqref{eq:Ibound2}}

Since the conditional mutual information is an average of unconditional mutual informations and \eqref{eq:d22} is uniform in $\veq$, it suffices to show that for any $P(x)$ and $Q(x)$ on some common alphabet $\Xv$, the inequality $d_{\TV}(P,Q) \le \delta$ implies $|I_P(X;Y) - I_Q(X;Y)| \le \delta\log\frac{4}{\delta}$.  Here the subscripts $P$ and $Q$ denote which distribution on $X$ is used, whereas the conditional distribution $W(y|x)$ of $Y$ given $X$ is the same in both cases.  We use similar notations for entropies, such as $H_P(Y)$ and $H_P(Y|X)$.

Since $I(X;Y) = H(Y) - H(Y|X)$, we have
\begin{multline}
    \big| I_P(X;Y) - I_Q(X;Y) \big| \\ \le \big| H_P(Y) - H_Q(Y) \big| + \big| H_P(Y|X) - H_Q(Y|X) \big|.
\end{multline}
For the second term, we follow \eqref{eq:Ibound_pf1}--\eqref{eq:Ibound_pf3} to deduce that 
\begin{equation}
    \big| H_P(Y|X) - H_Q(Y|X) \big| \le d_{\TV}(P, Q) \log 2.
\end{equation}
Moreover, the same reasoning along with the identities $P_Y(y) = \sum_{x} P_X(x)P_{Y|X}(y|x)$ and $P_{Y|X}(y|x) \le 1$  gives
\begin{equation}
    d_{\TV}(PW,QW) \le d_{\TV}(P, Q),
\end{equation}
where $PW$ denotes the $Y$-marginal of $P(x)W(y|x)$, and similarly for $QW$.  We may thus apply the result on the continuity of entropy in \cite[Ch.~2]{Csi11} to obtain
\begin{equation}
    \big| H_P(Y) - H_Q(Y) \big| \le d_{\TV}(P, Q) \log \frac{2}{d_{\TV}(P, Q)}. \label{eq:HY_bound}
\end{equation}
Combining the above estimates yields $|I_P(X;Y) - I_Q(X;Y)| \le \delta\log\frac{4}{\delta}$ whenever $d_{\TV}(P,Q) \le \delta$, as desired.

\begin{rem} \label{rem:log_factor}
    The logarithmic factor in \eqref{eq:HY_bound} can be replaced by a constant whenever $P$ and $Q$ yield probabilities of $Y = 0$ and $Y=1$ that are strictly bounded away from one.  This is because the entropy has bounded derivatives except as $P_Y(y)\to0$.  In fact, in the vicinity of $P_Y = \{0.5,0.5\}$ (which is relevant for symmetric settings), we may even make the bound in \eqref{eq:HY_bound} behave as $o(d_{\TV}(P, Q))$, since the derivative of the binary entropy function at $0.5$ is zero.
\end{rem}  

\section*{Acknowledgment}

This work was supported by the European Commission (ERC Future Proof), SNF (200021-146750 and CRSII2-147633), and `EPFL Fellows' program (Horizon2020 665667).

 \bibliographystyle{IEEEtran}
 \bibliography{../JS_References}

\end{document}